\documentclass[a4paper,11pt]{article}
\usepackage{amsmath}
\usepackage{amsthm}
\usepackage{amssymb}
\usepackage{mathpazo}
\usepackage{hyperref}
\hypersetup{colorlinks=true,linkcolor=black,citecolor=[rgb]{0.5,0,0}}
\usepackage{nicefrac}
\usepackage{cleveref}
\usepackage{tikz}
\usetikzlibrary{arrows.meta}

\usepackage[margin=1in]{geometry}
\usepackage{setspace}
\newtheorem{theorem}{Theorem}

\newtheorem{lemma}[theorem]{Lemma}

\newtheorem{remark}[theorem]{Remark}

\usepackage{xspace}

\newcommand{\PIH}{$\mathsf{PIH}$\xspace}
\newcommand{\CSP}{$\mathsf{CSP}$\xspace}
\newcommand{\W}{$\mathsf{W}$}
\newcommand{\E}{\mathbb{E}}
\newcommand{\poly}{\mathsf{poly}}
\newcommand{\PCP}{$\mathsf{PCP}$\xspace}
\newcommand{\FPT}{$\mathsf{FPT}$\xspace}
\newcommand{\NP}{$\mathsf{NP}$\xspace}

\newcommand{\calF}{\mathcal{F}}
\newcommand{\calC}{\mathcal{C}}

\newcommand{\U}{\mathcal{U}}

\newcommand{\CS}{\mathcal{S}}
\newcommand{\opt}{\mathsf{OPT}}
\newcommand{\ind}{\mathbf{1}}
\newcommand{\val}{\mathrm{val}}

\newcommand{\maxcover}{\mbox{\sf MaxCover}\xspace}

\newcommand{\supp}{\mathrm{supp}}
\newcommand{\improv}{\mathsf{improv}}
\newcommand{\cost}{\mathsf{cost}}
\renewcommand{\tilde}{\widetilde}

\newcommand{\eps}{\varepsilon}

\newcommand{\N}{\mathbb{N}}

\newcommand{\mc}{$k$-$\mathsf{maxcoverage}$\xspace}
\newcommand{\gapcsp}{\mathsf{Gap2CSP}\xspace}
\newcommand{\gapvcsp}{\mathsf{Gap2VCSP}\xspace}
\newcommand{\gapmaxcpv}{\mathsf{GapMax}$k$\mathsf{Cov}\xspace}
\newcommand{\gapmaxcpvbd}{\mathsf{GapMax}$k$\mathsf{CovSmallUni}\xspace}
\newcommand{\kmed}{{$k$-median}\xspace}

\newcommand{\calI}{\mathcal{I}}
\title{{ On Equivalence of  Parameterized Inapproximability of   $k$-Median, $k$-Max-Coverage, and  2-CSP}}
\date{}
\author{Karthik C.\ S.\thanks{This work was supported by the National Science Foundation under Grant CCF-2313372 and by the Simons  
 Foundation, Grant Number 825876, Awardee Thu D. Nguyen.}\vspace{0.1cm}\\\vspace{0.1cm} Rutgers University\\ \texttt{karthik.cs@rutgers.edu}\and Euiwoong Lee\footnote{Supported in part by NSF grant CCF-2236669 and Google.}\vspace{0.1cm}\\\vspace{0.1cm}  University of Michigan\\ \texttt{euiwoong@umich.edu}\and Pasin Manurangsi\thanks{Part of this work was done while the author was visiting Rutgers University.}\vspace{0.1cm}\\\vspace{0.1cm} Google Research\\ \texttt{pasin@google.com}}
\begin{document}
\maketitle

\begin{abstract}
    Parameterized Inapproximability Hypothesis (\PIH) is a central question in the field of parameterized complexity. \PIH asserts that given as input a 2-\CSP on $k$ variables and alphabet size $n$, it is \W$[1]$-hard parameterized by $k$ to distinguish if the input is perfectly satisfiable or if every assignment to the input violates 1\% of the constraints. \vspace{0.15cm}
    
    An important implication of \PIH is that it yields the tight parameterized inapproximability of the \mc problem. In the \mc problem, we are given as input a set system, a threshold $\tau>0$, and a parameter $k$ and the goal is to determine if there exist $k$ sets in the input whose union is at least $\tau$ fraction of the entire universe. \PIH is known to imply that it is  \W$[1]$-hard parameterized by $k$ to distinguish if there are $k$ input sets whose union is at least $\tau$ fraction of the universe or if the union of every $k$ input sets is not much larger than $\tau\cdot (1-\frac{1}{e})$ fraction of the universe. \vspace{0.15cm}

    In this work we present a gap preserving \FPT reduction (in the reverse direction) from the \mc problem to the aforementioned 2-\CSP problem, thus showing that the assertion that  approximating the \mc problem to some constant factor is \W$[1]$-hard implies \PIH.  In addition, we present a gap preserving \FPT reduction from the $k$-median problem (in general metrics) to the \mc problem, further highlighting the power of gap preserving \FPT reductions over classical gap preserving polynomial time reductions. 
    
\end{abstract}
\clearpage
\section{Introduction}
Approximation Algorithms and Fixed Parameter Tractability are two popular ways to cope with \NP-hardness of computational problems. In recent years, there is a 
steady rise in results contributing to the theory of parameterized inapproximability (see e.g.~\cite{FKLM20}). These include parameterized inapproximability results for the $k$-Set Intersection problem \cite{Lin18,BukhSN21}, $k$-Set Cover problem \cite{CL19,KLM19,L19,KN21,LinRSW23}, $k$-Clique problem \cite{lin2021constant,KK22,CFLL23}, Steiner Orientation problem \cite{Wlo19}, fundamental problems in coding theory and lattice theory  \cite{BBEGKLMM21,BennettCG023}, and more. 

A major open question in the theory of parameterized inapproximability is the resolution of the \emph{parameterized inapproximability hypothesis} (\PIH) \cite{LRSZ20}. \PIH asserts   that there exists some $\varepsilon>0$, such that it is \W$[1]$-hard to distinguish if a 2-\CSP instance on $k$ variables ($k$ is the parameter) and alphabet size $n$ is completely satisfiable or if every assignment to its variables satisfies at most $1-\varepsilon$ fraction of the constraints. The analogue of (the resolution of) \PIH in the \NP-world is the celebrated \PCP theorem~\cite{AS98,ALMSS98,D07}, and thus positively resolving \PIH is also appropriately dubbed as proving the ``\PCP theorem for Parameterized Complexity'' \cite{LRSZ20}.

\PIH is known to imply the hardness of approximation (to some positive constant factor) of many fundamental problems in parameterized complexity (for which we do not yet know how to prove unconditional constant factor \W$[1]$-hardness), such as \mc \cite{CGLL19} (and consequently clustering problems such as minimizing the $k$-median and $k$-means objectives; see \cite{CGLL19}), Directed Odd Cycle Transversal \cite{LRSZ20}, Strongly Connected Steiner Subgraph \cite{CFM22},  Planar Steiner Orientation \cite{ChitnisFS19}, Grid Tiling and Determinant Maximization \cite{ohsaka2022parameterized}, Independent Set in $H$-free graphs \cite{dvovrak2022parameterized}, etc.  

It is known that assuming the Gap Exponential Time Hypothesis (\textsf{Gap-ETH}) \cite{MR16,D16} one can show that the gap problem on 2-\CSP referred to in \PIH is not in \FPT (see e.g. \cite{BGKM18}), and moreover that many parameterized inapproximability results discussed in this paper follow from \textsf{Gap-ETH} (e.g. \cite{CCKLM20,Manurangsi20}). Moreover, in a recent breakthrough, Guruswami et al. \cite{GLRSW24}, building on ideas in \cite{LinRSW23FOCS}, proved that assuming the Exponential Time Hypothesis \cite{IP01,IPZ01}, the aforementioned gap problem on 2-\CSP is not in \FPT.  Subsequently, they even obtained near optimal conditional time lower bounds for the parameterized 2-\CSP problem \cite{GLRSWpreprint24}. However, this paper is solely focused on parameterized inapproximability, and thus we will not further elaborate on the related works in fine-grained inapproximability. 

A popular approach to make progress on central questions such as resolving \PIH, is to prove unconditional hardness results for problems whose hardness is known only assuming \PIH. Such an approach (of proving the implications unconditionally) has been historically very fruitful in complexity theory, for example, in the last decade it is in the attempt of proving improved unconditional \NP-hardness of approximation result for the vertex cover problem \cite{KMS17} (whose optimal inapproximability is only known under Unique Games Conjecture \cite{khot2002power}), that the 2-to-2 games theorem was proven \cite{KMS18}. Moreover, this approach is indeed an active line of research in the theory of parameterized inapproximabily, leading  to the hardness of approximation results   for $k$-set cover problem \cite{CL19,KLM19,L19}, $k$-clique problem \cite{lin2021constant,KK22,CFLL23}, and more.

The maximization version of the $k$-set cover problem, i.e., the \mc problem is a fundamental optimization problem at the heart of  many computation problems in computer science. For example, it is at the heart of many clustering problems \cite{guha1999greedy}. Formally, in the  \mc problem we are given as input a pair $(\mathcal{U},\mathcal{S})$ and a parameter $k$, where $\mathcal{S}$ is a collection of sets over the universe $\mathcal{U}$, and the goal is to find $k$ sets in $\mathcal{S}$ whose union is of maximum cardinality. The  \mc problem is a canonical \W$[2]$-complete problem, and currently the \W$[1]$-hardness of approximating the \mc problem to some constant factor only holds assuming \PIH \cite{CGLL19}. Thus, we ask:

\begin{center}
    \textit{Is it possible to prove constant inapproximability of the \mc problem\\ circumventing the resolution of \PIH?}
\end{center}

This question is particularly appealing since \W$[1]$-hardness of approximating $k$-set cover (minimization variant of \mc) was established circumventing \PIH \cite{CL19,KLM19,L19}, and more recently,  similar progress was achieved for the $k$-clique problem as well \cite{lin2021constant,KK22,CFLL23}. 

Our first result is rather surprising (at least in first glance),  that the answer to the above question is in the negative.

\begin{theorem}[Informal statement; See Theorem~\ref{thm:mctocspformal} for a formal statement]\label{thm:mctocsp}
For every $\delta \in (0,1/2]$ (where $\delta$ is allowed to depend on $k$), if approximating the \mc problem to $(1-\delta)$ factor is \W$[1]$-hard then approximating 2-\CSP to $\left(1-\frac{\delta^2}{4}\right)$ factor is also \W$[1]$-hard (under randomized Turing reductions).
\end{theorem}

Recall that in \cite{CGLL19}, the authors proved that for every $\varepsilon>0$, assuming \PIH, approximating the \mc problem to a factor better than $1-\frac{1}{e}+\varepsilon$ is \W$[1]$-hard  (see Lemma~19 in \cite{CGLL19}). Moreover, this inapproximability factor is tight \cite{hochba1997approximation}.  Together with \Cref{thm:mctocsp}, we have that there is a gap preserving reduction in both directions\footnote{Albeit the reduction from the \mc problem to the 2-\CSP problem is a randomized Turing reduction.} between 2-\CSP (on $k$ variables and alphabet size $n$) and the \mc problem on $\poly(n)$ sets over universe of size $\poly(n)$.

To the best of our knowledge, we are not aware of a direct gap preserving reduction to the 2-\CSP problem from the \textsf{maxcoverage} problem in the \NP world, i.e., in other words we do not know how to directly prove the \PCP theorem for \NP, assuming that approximating \textsf{maxcoverage} to some constant factor is \NP-hard.  Furthermore, an interesting consequence of Theorem~\ref{thm:mctocsp} is a gap amplification result for the \mc problem in the parameterized complexity world: starting from the \W$[1]$-hardness of approximating the \mc problem to $(1-\varepsilon)$ factor (for some constant $\varepsilon>0$), we can first apply Theorem~\ref{thm:mctocsp} and then Lemma~19 of \cite{CGLL19}, to obtain that approximating the \mc problem to $(1-\frac{1}{e}+\varepsilon')$ factor, for any $\varepsilon'>0$ is \W$[1]$-hard.

At a high level, the proof of Theorem~\ref{thm:mctocsp} has three steps. Starting from an instance of the \mc problem containing $n$ sets, in the first step we subsample an universe of size $O_k(\log n)$ while retaining the gap in the completeness and soundness cases (Lemma~\ref{lem:uni-red}). In the second step, we reduce from this new \mc instance on the smaller universe to a variant of the 2-\CSP instance called ``Valued \CSP'' (see Section~\ref{sec:prelim} for the definition) by first equipartitioning the universe into $O_k(1)$ many subuniverses and then constructing a Valued \CSP instance where for a subset of variables, each variable in that subset  is associated with a subuniverse and an assignment to that variable determines how each of the $k$ solution sets cover this subuniverse. The rest of the variables ($k$ many of them) encode the $k$ solution sets (see Lemma~\ref{lem:cov-to-vcsp}). Finally, in the last step, we provide a gap preserving reduction from Valued \CSP to the standard 2-\CSP (Lemma~\ref{lem:vcsp-to-csp}).

Theorem~\ref{thm:mctocsp} has further implications on our understanding of the complexity of approximating the \mc problem. While exactly solving the problem is \W$[2]$-hard, it was known to experts that approximating  the \mc problem to $1-\frac{1}{F(k)}$ factor, for any computable function $F$, is in \W$[1]$. A corollary of Theorem~\ref{thm:mctocsp} is a formal proof of this \W$[1]$-membership (by setting $\delta=1/F(k)$ in Theorem~\ref{thm:mctocsp}). Moreover,  by modifying the range of parameters in the reduction of \cite{KLM19} for the $k$-set cover instance, it is possible to argue that approximating  the \mc problem to any $1-1/\rho(k) $ factor is \W[1]-hard for every unbounded computable function $\rho$ (for example think of $\rho(k)=\log^*(k)$). Thus, we obtain the following. 

\begin{theorem}\label{thm:mccomplete}
    Let $\rho:\N\to\N$ be any unbounded computable function. Then approximating the \mc problem to $1-\frac{1}{\rho(k)}$ factor is \W$[1]$-complete.
\end{theorem}

We can extend our line of inquiry and wonder if one can prove the parameterized inapproximability of clustering objectives such as $k$-median and $k$-means without proving anything about the inapproximability of the \mc problem. We will restrict our attention here to the $k$-median problem, but our results extend to the $k$-means problem as well (see Remark~\ref{rem:kmean}). 

In the $k$-median problem (in general metric), we are given as input a tuple $((V,d),C,\mathcal{F},\tau)$ and a parameter $k$, where $V$ is a finite set, and $d$ is a distance function for all pairs of points in $V$ respecting the triangle inequality (or more precisely $(V,d)$ is a metric space), $C,\mathcal{F}\subseteq V$ and the goal is to determine if there exists $F\subseteq \mathcal{F}$ such that $|F|=k$ and $\cost(C,F)$ is minimized, where $\cost(C,F)$ is the sum of distances from every client in $C$ to its closest facility in $F$  (see Section~\ref{sec:prelim} for a formal definition). The  $k$-median problem is \W$[2]$-complete by a simple reduction from the \mc problem, and currently the \W$[1]$-hardness of approximating $k$-median to some constant factor only holds assuming \PIH \cite{CGLL19}. Thus, we ask: 

\begin{center}
    \textit{Is it possible to prove constant inapproximability of $k$-median\\ circumventing the resolution of \PIH? }
\end{center}

Our result again is rather surprising that the answer to the above question is also in the negative.

\begin{theorem}\label{thm:kmedtomc}For any constants $\alpha,\delta >0$, if approximating the $k$-median problem to $(1+\alpha+\delta)$ factor is \W$[1]$-hard then approximating the multicolored \mc problem to $\left(1-\frac{\alpha}{2}\right)$ factor is also \W$[1]$-hard (under randomized Turing reductions). A similar reduction also holds from the $k$-means problem to the multicolored \mc problem.
\end{theorem}

In Remark~\ref{rem:color} we discuss how to modify the proof of Theorem~\ref{thm:mctocsp} to obtain the statement of Theorem~\ref{thm:mctocsp} to hold even for the multicolored \mc problem instead of the (unmulticolored) \mc problem. Then we can put together the reduction in Theorem~\ref{thm:kmedtomc} with the modified Theorem~\ref{thm:mctocsp} to obtain an \FPT gap preserving reduction from the $k$-median problem to 2-\CSP problem.

The proof of Theorem~\ref{thm:kmedtomc} follows from observing that the $\left(1+\frac{2}{e}+\varepsilon\right)$-approximation algorithm (for any $\varepsilon>0$) of \cite{CGLL19} can be fine-tuned and rephrased as a reduction from $k$-median problem to the multicolored \mc problem. 

Both Theorems~\ref{thm:mctocsp} and \ref{thm:kmedtomc} illustrate the power of \FPT gap preserving reductions over classical gap preserving polynomial time reductions. In particular, if we had a polynomial time analogue of Theorem~\ref{thm:kmedtomc}, i.e., if the runtime of the algorithm $\mathcal{A}$ and $\Gamma$ in the Theorem~\ref{thm:kmedtomc} statement are both polynomial functions then  this would lead to  a major breakthrough in the field of approximation algorithms. In particular, it would show optimal approximation thresholds for the celebrated $k$-median and $k$-means problems in the \NP world, improving on the state-of-the-art result of \cite{cohen2023breaching} and \cite{kanungo2004local} respectively, showing that the hardness of approximation factors obtained in \cite{guha1999greedy} are optimal!

In Figure~\ref{fig:enter-label}, we highlight gap-preserving \FPT reductions between $k$-clique, 2-\CSP, \mc, and $k$-median and $k$-means problems. 
A glaring open problem in Figure~\ref{fig:enter-label} is whether constant inapproximability of the $k$-clique problem implies \PIH. This is a challenging open problem, even listed in \cite{FKLM20}.

\begin{figure}
    \resizebox{\textwidth}{!}{\begin{tikzpicture}
\small
\node (kclique) [draw=red!80!black,thick] at (1.5, 0) {$k$-\textsc{Clique}};

\node (maxcoverage) [draw=red!80!black,thick] at (6.75,3) {$k$-\textsc{Max-Coverage}};

\node (2-CSP) [draw=red!80!black,thick] at (1.5, 3) {2-\textsc{CSP}};

\node (kmedian) [draw=red!80!black,thick,text width=2cm] at (12.5, 3) {\centering\ $k$-\textsc{Median} \& $k$-\textsc{Means}};

 \begin{scope}[transform canvas={xshift=.4em}]
 \draw [-{Latex[length=1.5mm, width=1.5mm]}] (kclique) -- (2-CSP); 
\end{scope}

 \begin{scope}[transform canvas={xshift=-.4em}]
 \draw [-{Latex[length=1.5mm, width=1.5mm]}] (2-CSP) -- (kclique); 

\end{scope}

 \begin{scope}[transform canvas={yshift=-0.4em}]
 \draw [-{Latex[length=1.5mm, width=1.5mm]}] (2-CSP) -- (maxcoverage);

\draw [-{Latex[length=1.5mm, width=1.5mm]}] (maxcoverage) -- (kmedian);
\end{scope}

\begin{scope}[transform canvas={yshift=0.4em}]
 \draw [-{Latex[length=1.5mm, width=1.5mm]}] (maxcoverage) -- (2-CSP) ;

\draw [dashed, -{Latex[length=1.5mm, width=1.5mm]}] (kmedian) -- (maxcoverage);
\end{scope}

\node [above, align=center] at (3.5, 3.1) 
{\footnotesize Theorem~\ref{thm:mctocsp}};
\node [above, align=center] at (3.5, 2.35) 
{\footnotesize \cite{feige1998threshold}};

\node [above, align=center] at (2.1, 1.3) 
{\footnotesize {\color{blue}Open}}; 
\node [above, align=center] at (0.6, 1.35) {\footnotesize \cite{FeigeGLSS96}}; 

\node [above, align=center] at (9.8, 3.1) 
{\footnotesize Theorem~\ref{thm:kmedtomc}};
 \node [above, align=center] at (9.8, 2.35) {\footnotesize \cite{guha1999greedy}};

 \node [above, align=center] at (9.25, 0) {\huge Gap Preserving Reductions};

\end{tikzpicture}}
\caption{In the above figure, we provide bidirectional \FPT gap preserving reductions between $k$-clique, 2-\CSP, \mc, and $k$-median and $k$-means problems, whenever possible, with appropriate references. The reduction from $k$-median and $k$-means problems is to the multicolored version of the \mc problem (see Remark~\ref{rem:mc} for a discussion) and that is why the arrow is dashed. }
    \label{fig:enter-label} 
\end{figure}

\section{Preliminaries}\label{sec:prelim}

In this section, we formally define the problems of interest to this paper. Throughout, we use the notation $O_k(\cdot)$ and $\Omega_k(\cdot)$ to denote that the hidden constant can be any computable function of $k$. 

\paragraph{\mc problem.} We denote by $(\U,\CS)$ a set system where $\U$ denotes the universe and $\CS$ is a collection of subsets of $\U$. In the \mc problem, we are given as input a set system $(\U,\CS)$ and a parameter $k$, and the goal is to identify $k$ sets in $\CS$ whose union is of maximum size. We denote by $\opt(\U, \CS)$ the optimum fraction of the \mc instance $(\U,\CS)$, i.e. $\underset{S_{i_1}, \ldots, S_{i_k} \in \CS}{\max} \frac{|S_{i_1} \cup \cdots \cup S_{i_k}|}{|\U|}$.

\begin{sloppypar}
    We denote by  $\gapmaxcpv(\tau, \tau')$ the decision problem where given as input an instance $(\U,\CS)$ of the \mc problem, the goal is to distinguish the completeness case where $\opt(\U, \CS)\ge\tau$ and the soundness case where $\opt(\U, \CS)<\tau'$. We also define $\gapmaxcpvbd(\tau, \tau')$ be the same problem but only restricted to the instances where $|\U| \leq O_k(\log |\CS|)$.
\end{sloppypar}

In this paper, we also refer to the multicolored \mc problem whose input is $(\mathcal{U},\mathcal{S}:=\mathcal{S}_1\dot\cup \mathcal{S}_2\dot\cup \cdots \dot\cup \mathcal{S}_k)$, and the goal is to identify $(S_1,S_2,\ldots ,S_k)\in \mathcal{S}_1\times \mathcal{S}_2\times \cdots \times \mathcal{S}_k$ such that $|S_{1} \cup S_2\cup \cdots \cup S_{k}|$ is maximized. Moreover, we extend the above notations of $\opt$, $\gapmaxcpv$, and $\gapmaxcpvbd$ to the multicolored \mc problem.

\paragraph{2-\CSP.} For convenience, we use weighted version of 2-\CSP where the edges are weighted. Note that there is a simple (\FPT) reduction from this version to the unweighted version~\cite{CrescenziST01}.

A 2-\CSP instance $\Pi = (V, E, (\Sigma_v)_{v \in V}, (w_e)_{e \in E}, (C_e)_{e \in E})$ consists of the following:
\begin{itemize}
\item The set of vertices (i.e. variables) $V$.
\item The set $E$ of edges between $V$.
\item For each $v \in V$, the alphabet set $\Sigma_v$ of $v$.
\item For each $e = (u, v) \in E$, a weight $w_e$ and the constraint $C_e \subseteq \Sigma_u \times \Sigma_v$.
\end{itemize}

An \emph{assignment} is a tuple $\psi = (\psi_v)_{v \in V}$ where $\psi_v \in \Sigma_v$. The \emph{(weighted and normalized) value} of an assignment $\psi$, denoted by $\val_{\Pi}(\psi)$, is defined as: $$\frac{1}{\underset{e \in E}{\sum} w_e} \cdot \underset{e = (u, v) \in E}{\sum} w_e \cdot \ind[(\psi_u, \psi_v) \in C_{u, v}],$$
where for any proposition $\Lambda$, $\ind(\Lambda)$ is 1 if $\Lambda$ is true and 0 otherwise. The value of the instance is defined as $\val(\Pi) := \underset{\psi}{\max}\ \val_{\Pi}(\psi)$ where the maximum is over all assignments $\psi$ of $\Pi$.
The \emph{alphabet size} of the instance is defined as $\underset{v \in V}{\max}\ |\Sigma_v|$. If a 2-\CSP instance is provided without $(w_e)_{e\in E}$ then the weights are all assumed to be 1. 

The $\gapcsp(c, s)$ problem is to decide whether a 2-\CSP instance $\Pi$ has value at least $c$ or less than $s$. Note that the parameter of this problem is the number of variables in $\Pi$.

\paragraph{(Finite) Valued 2-\CSP.} It will also be more convenient for us to employ a more general version known as \emph{(Finite) Valued \CSP}\footnote{See e.g. \cite{KolmogorovKR15} and references therein.}. In short, this is a version where different assignments for each edge can results in different values. This is defined more precisely below. 

A Valued 2-\CSP instance $\Pi = (V, E, (\Sigma_v)_{v \in V}, (f_e)_{e \in E})$ consists of
\begin{itemize}
\item $V, E, (\Sigma_v)_{v \in V}$ are defined similarly to 2-\CSP instances.
\item For each $e = (u, v) \in E$, the value function $f_e: \Sigma_u \times \Sigma_v \to [0, 1]$.
\end{itemize}
The notion of an assignment is defined similar to 2-\CSP instance but the value of an assignment is now defined as:
\begin{align*}
\val_{\Pi}(\psi) := \underset{(u, v) \sim E}{\E}[f_{(u, v)}(\psi_u, \psi_v)]. 
\end{align*}
The value of an instance is defined similar to before.

The $\gapvcsp(c, s)$ problem is to decide whether a Valued 2CSP instance $\Pi$ has value at least $c$ or less than $s$. Again, the parameter here is the number of variables.

\paragraph{$k$-median.} An instance of the \kmed problem is defined
by a tuple $((V,d), C, \calF, k)$, where $(V,d)$ is a metric space over a
set of points $V$ with $d(i,j)$ denoting the distance between two points
$i,j$ in $V$. Further, $C$ and $\calF$ are subsets of $V$ and are referred to
as ``clients'' and ``facility locations'' respectively, and $k$ is a positive
parameter.  The goal is to find a subset $F$ of $k$ facilities in $\calF$
to minimize
$$\cost(C, F) := \sum_{j \in C} d(j,F), $$ 
where $d(j,F):=\underset{f\in F}{\min}\ d(j,f)$.
The cost of the $k$-means objective is $\cost_2(C, F) := \sum_{j \in C} d(j,F)^2$.

\paragraph{$k$-\textsf{MaxCover} problem.} We recall the \maxcover problem introduced in \cite{CCKLM20}.
A $k$-\maxcover instance $\Gamma$ consists of a bipartite graph $G=(V\dot\cup W, E)$ such that $V$ is partitioned into $V=V_1\dot\cup \cdots \dot\cup V_k$ and $W$ is partitioned into $W=W_1\dot\cup \cdots \dot\cup W_\ell$. We sometimes refer to $V_{i}$'s and $W_j$'s as {\em left super-nodes} and {\em right super-nodes} of $\Gamma$, respectively. 

A solution to $k$-\maxcover is called a {\em labeling},
which is a subset of vertices $v_1\in V_1,\ldots ,v_k\in V_k$.
We say that a labeling $v_1,\ldots ,v_k$ {\em covers} a right super-node $W_{i}$, if there exists a vertex $w_{i} \in W_{i}$ which is a joint neighbor of all $v_1,\ldots ,v_k$, i.e., $(v_j,w_i)\in E$ for every $j\in[k]$.
We denote by $\maxcover(\Gamma)$ the maximal fraction of right super-nodes that can be simultaneously covered, i.e.,
\begin{align*}
\maxcover(\Gamma) = \frac{1}{\ell} \left(\max_{\text{labeling } v_1,\ldots ,v_k} \bigl\lvert\bigl\{i \in [\ell] \mid W_i \text{ is covered by } v_1,\ldots ,v_k\bigr\}\bigr\rvert\right). 
\end{align*}

Given an instance $\Gamma(G,c,s)$ of the  $k$-\maxcover problem as input,  our goal is to distinguish between the two cases:
\begin{description}
	\item[Completeness] $\maxcover(\Gamma)\ge c$.   
	\item[Soundness] $\maxcover(\Gamma)\le s$.  
\end{description}

\paragraph{Concentration Inequalities.} We will also use the multiplicative Chernoff inequality, which is summarized below.
\begin{theorem}[Chernoff Inequality] \label{thm:chernoff}
Let $X_1, \ldots, X_m$ denote i.i.d. Bernoulli random variables where $\E[X_j] = q$. Then, for any $\zeta \in (0, 1)$ we have
\begin{align*}
\Pr[X_1 + \cdots + X_m \geq (1 + \zeta) q m], \Pr[X_1 + \cdots + X_m \leq (1 - \zeta) q m] \leq \exp\left(-\zeta^2 
q m / 3\right).
\end{align*}
\end{theorem}
\section{Reducing  $k$-MaxCoverage to 2-\CSP}

In this section, we prove the following formal version of Theorem~\ref{thm:mctocsp}. 

\begin{theorem}\label{thm:mctocspformal}
For every $\tau>0$ and   $\delta \in (0,1/2]$ (where both $\tau$ and $\delta$ are allowed to depend on $k$), there is a randomized algorithm $\mathcal{A}$ which takes as input a \mc instance $(\mathcal{U},\mathcal{S})$ and with probability $1-o(1)$, outputs  $\Gamma(k)$ many 2-\CSP instances $\{\Pi_i\}_{i\in [\Gamma(k)]}$, for some computable function $\Gamma:\N\to\N$, such that the following holds.  
\begin{description}
    \item[Running Time:] $\mathcal{A}$ runs in time $T(k)\cdot \poly(|\mathcal{U}|+|\mathcal{S}|)$, for some computable function $T:\N\to\N$. 
    \item[Size:] For every $i\in[\Gamma(k)]$, we have that  $\Pi_i$ is defined on $\Lambda(k)$ variables over an alphabet of size $\poly(|\mathcal{U}|+|\mathcal{S}|)$ for some computable function $\Lambda:\N\to\N$. 
    \item[Completeness:] Suppose there exist $k$ sets in $\mathcal{S}$ such that their union is of size $\tau \cdot |\mathcal{U}|$. Then, there exists $i\in [\Gamma(k)]$ and an assignment to the variables of $\Pi_i$ that satisfies all its constraints. 
    \item[Soundness:]  Suppose that for every $k$ sets in $\mathcal{S}$, their union is of size at most $(1-\delta)\cdot \tau \cdot |\mathcal{U}|$. Then, for every $i\in [\Gamma(k)]$ we have that every assignment to the variables of $\Pi_i$ satisfies at most $\left(1-\frac{\delta^2}{4}\right)$ fraction of the constraints of $\Pi_i$.
\end{description}
\end{theorem}

The proof of the theorem follows in three steps. 
We start by providing a randomized reduction which shows that we may assume w.l.o.g. that $|\U| \leq O_k(\log n)$. The rough idea is to use random hashing and subsampling to reduce the domain, as formalized below.

\begin{lemma} \label{lem:uni-red}
For every $\tau>0$ and   $\delta \in (0,1/2]$ (where both $\tau,\delta$ may or may not depend on $k$), there is a randomized \FPT reduction (that holds w.p.\ $1-o(1)$) from $\gapmaxcpv(\tau, (1 - \delta)\tau)$ to $\gapmaxcpvbd(\tau', (1 - \eps)\tau')$ for $\eps = \delta^2/2$ and $\tau'=\delta (1-\delta)\cdot (1+\delta^2)$.
\end{lemma}

In the second step, we show how to reduce the small-universe \mc instance to a Valued \CSP instance. The overall idea is to create a set of variables $x_1, \ldots, x_k$ where $x_i$ represents the $i$-th set selected in the solution. To check that they cover a large number of constraints, we partition the universe into $M$ groups $\U_1, \ldots, \U_M$ each of size $O(\log n / \log k)$ where the small-universe \mc instance gurantees that $M=O_k(1)$. For each partition $j \in [M]$, we create a variable $y_j$. The variable encodes how $x_1, \ldots, x_k$ covers the $j$-th partition $\U_j$. Namely, each $\sigma \in \Sigma_{y_j}$ encodes whether each element (in $\U_j$) is covered and, if so, by which set. Notice that there can be as many as $(k + 1)^{|\U_j|}$ possibilities here, but this is not an issue since $|\U_j| = O(\log n / \log k)$. The constraints are then simply the consistency checks between $x_i, y_j$ where the values represents the number of elements the $i$-th set covers in $\U_j$. Formally, we prove the following.

\begin{lemma} \label{lem:cov-to-vcsp}
\begin{sloppypar}For every $\tau,\eps>0$  (where both $\tau,\eps$ may or may not depend on $k$), there is a deterministic \FPT reduction from $\gapmaxcpvbd(\tau, (1 - \eps)\tau)$ to $\gapvcsp(c, c(1 - \eps))$ where $c = O_k(\tau)$.\end{sloppypar}
\end{lemma}

The final part is the following lemma which shows that, in the \FPT world, $\gapvcsp$ reduces to $\gapcsp$. At a high level, this reduction is done by guessing the values of each edge (in the optimal solution) and turning that into a ``hard'' constraint as in a 2-\CSP. 

\begin{lemma} \label{lem:vcsp-to-csp}
For any $c > s > 0$ such that $c/s \geq 1 + \Omega_k(1)$, there is a deterministic \FPT (Turing) reduction from $\gapvcsp(c, s)$ to $\gapcsp(1, 1 - \eps)$ for $\eps = \frac{c/s - 1}{c/s + 1}$.
\end{lemma}

Finally, we put together the above three lemmas to prove Theorem~\ref{thm:mctocspformal}.

\begin{proof}[Proof of Theorem~\ref{thm:mctocspformal}]
    The algorithm $\mathcal{A}$ takes as input an instance of $\gapmaxcpv(\tau, (1 - \delta)\tau)$, applies the reduction in Lemma~\ref{lem:uni-red} to obtain an instance of $\gapmaxcpvbd(\tau', (1 - \frac{\delta^2}{2})\tau')$ for $\tau'=\delta (1-\delta)\cdot (1+\delta^2)$, and then applies the reduction in Lemma~\ref{lem:cov-to-vcsp} to obtain an instance of  $\gapvcsp(c, c(1 - \frac{\delta^2}{2}))$ where $c = O_k(\delta)$, and finally applies the reduction in Lemma~\ref{lem:vcsp-to-csp} to obtain an instance of $\gapcsp(1, 1 - \eps)$ for $\eps = \frac{\delta^2}{4-\delta^2}>\frac{\delta^2}{4}$. 
\end{proof}
\subsection{Step I: Universe Reduction for $k$-MaxCoverage}

In this subsection, we prove Lemma~\ref{lem:uni-red}. 

\begin{proof}[Proof of Lemma~\ref{lem:uni-red}]
Given $( \U,\CS = \{S_1, \ldots, S_n\})$, an  instance of the $\gapmaxcpv(\tau, (1 - \delta)\tau)$ problem, we create an instance $( \U',\CS' = \{S'_1, \ldots, S'_n\})$ of the $\gapmaxcpvbd(\tau', (1 - \eps)\tau')$ problem as follows.
Let $m := \lceil 12 k \cdot \delta^{-9} \log n  \rceil$ and $p := \frac{\delta}{\tau\cdot |\U|}$. Let $\U' = [m]$.
For each $u \in \U$ and $j \in [m]$, let $Y_{u, j}$ denote an i.i.d. Bernoulli random variable that is 1 with probability $p$. Then, for each $i \in [n]$ and $j \in [m]$, let $j$ belong to $S'_i$ if and only if there exist some $u\in S_i$ such that $Y_{u, j} = 1$. 

Fix any $S_{i_1}, \ldots, S_{i_k} \in \CS$. For every $j\in[m]$, let $X_j$ denote the indicator whether $j \in S'_{i_1} \cup \cdots \cup S'_{i_k}$. Note that $X_1, \ldots, X_m$ are i.i.d. and,
\begin{align*}
\Pr[X_j = 1] = 1 - (1 - p)^{|S_{i_1} \cup \cdots \cup S_{i_k}|}.
\end{align*}
Note also that $X_1 + \cdots + X_m$ is exactly equal to $|S'_{i_1} \cup \cdots \cup S'_{i_k}|$.

\paragraph{Completeness.} Suppose that $\opt(\U, \CS) \geq \tau$. Let $S_{i_1}, \ldots, S_{i_k}$ be an optimal solution in $(\U, \CS)$. Then, we have for each $j\in[m]$ that:
\begin{align*}
\Pr[X_j = 1] \geq 1 - (1 - p)^{\tau\cdot |\U|} \geq 1 - \frac{1}{(1 + p)^{\tau\cdot |\U|}} \geq 1 - \frac{1}{1 + p\tau|\U|} = \frac{\delta}{1+\delta}=\frac{1}{1-\delta^4}\cdot \tau',
\end{align*}
where the second inequality follows from Bernoulli's inequality and the last inequality follows from our choice of parameters. Applying \Cref{thm:chernoff} with $\zeta=\delta^4$ implies that $\Pr[X_1 + \cdots +X_m \geq \tau'] \geq 1 - \exp(-\delta^9 m/6)=1-o(1)$ as desired.

\paragraph{Soundness.} Suppose that $\opt(\U, \CS) < (1 - \delta)\tau$. Consider any $S_{i_1}, \ldots, S_{i_k} \in \CS$. We have for each $j\in[m]$ that:
\begin{align*}
\Pr[X_j = 1] \leq 1 - (1 - p)^{(1-\delta)\tau\cdot |\U|} \leq (1 - \delta)p\tau|\U| =\delta (1-\delta)=\frac{1}{1+\delta^2}\cdot \tau', 
\end{align*}
where the second inequality follows from Bernoulli's inequality. 
Again, applying \Cref{thm:chernoff} with $\zeta=\delta^4$ implies that $\Pr[X_1 + \cdots +X_m \geq \frac{1+\delta^4}{1+\delta^2} \tau'] \leq \exp(-\delta^9(1-\delta)m/3)\le \exp(-\delta^9m/6)$ (where we used that $\delta\le 1/2$). Taking the union bound over all $i_1, \ldots, i_k$ then implies that this holds for all $i_1, \ldots, i_k$ with probability at least $1 - \frac{n^k}{\exp(\delta^9m/6)}\geq 1-\frac{1}{n^k}=1-o(1)$.
\end{proof}

\begin{remark}\label{rem:color0}
    We note that we can mimic the above proof to extend Lemma~\ref{lem:uni-red}  to the multicolored \mc problem as well. In particular, this gives a randomized \FPT reduction (that holds w.p.\ $1-o(1)$) from multicolored $\gapmaxcpv(\tau, (1 - \delta)\tau)$ to multicolored $\gapmaxcpvbd(\tau', (1 - \eps)\tau')$ for $\eps = \delta^2/2$ and $\tau'=\delta (1-\delta)\cdot (1+\delta^2)$ as well.
\end{remark}

\subsection{Step II: Small-Universe $k$-MaxCoverage $\Rightarrow$ Valued \CSP}

In this subsection, we prove Lemma~\ref{lem:cov-to-vcsp}. 

\begin{proof}[Proof of Lemma~\ref{lem:cov-to-vcsp}]
Given an instance $( \U,\CS = \{S_1, \ldots, S_n\})$ of $\gapmaxcpvbd(\tau, (1 - \eps)\tau)$, we construct an instance $\Pi = (V, E, (\Sigma_v)_{v \in V}, (f_e)_{e \in E})$ of $\gapvcsp(c, c(1 - \eps))$ as follows:
\begin{itemize}
\item Let $M := \left\lceil \frac{|\U|}{\log |\CS|}\cdot  \log k \right\rceil=O_k(1)$ and let $\U_1 \dot\cup \cdots\dot\cup \U_M$ be a partition of $\U$ into nearly equal parts, each of size $|\U_i| = O(|\U| / M) = O(\log n / \log k)$.
\item Let $V = \{x_1, \ldots, x_k, y_1, \ldots, y_M\}$ and $E$ contains $(x_i, y_j)$ for all $i \in [k]$ and $j \in [M]$.
\item For each $i \in [k]$, let $\Sigma_{x_i} = [n]$.
\item For each $j \in [M]$, let $\Sigma_{y_j}$ contains all functions from $\U_j$ to $\{0, \ldots, k\}$.
\item For each $i \in [k], j \in [M]$, let $f_{(x_i, y_j)}$ be defined as follows:
\begin{align*}
f_{(x_i, y_j)}(\sigma_u, \sigma_v) =
\begin{cases}
\frac{|\sigma_v^{-1}(i)|}{|\U|} & \text{ if } \sigma_v^{-1}(i) \subseteq S_{\sigma_u}, \\
0 & \text{ otherwise.} 
\end{cases}
\end{align*}
\item Finally, let $c = \frac{\tau}{k \cdot M }$.
\end{itemize}

Note that the new parameter is $k + M = O_k(1)$. Furthermore, the running time of the reduction is polynomial since $|\Sigma_{y_j}| = (k+1)^{|\U_j|} = k^{O(\log n / \log k)} = n^{O(1)}$. Thus, the reduction is an \FPT reduction as desired.

We next prove the completeness and soundness of the reduction. In fact, we will argue that $\val(\Pi) = \frac{\opt(\U, \CS)}{k \cdot M }$, from which the completeness and soundness immediately follow. To see that $\val(\Pi) \geq \frac{\opt(\U, \CS)}{k \cdot M }$, let $S_{\ell_1}, \ldots, S_{\ell_k}$ denote an optimal solution. We let $\psi_{x_i} = \ell_i$ for all $i \in [k]$. As for $\psi_{y_j}$, we let $\psi_{y_j}(u)$ be 0 if $u \notin S_{\ell_1} \cup \cdots \cup S_{\ell_k}$; otherwise, we let $\psi_{y_j}(u) = i$ such that $S_{\ell_i}$ (if there are multiple such $i$'s, just pick one arbitrarily). It is obvious by the construction that $\psi_{y_j}^{-1}(i) \subseteq S_{\psi_{x_i}}$ for all $i \in [k]$ and $j \in [M]$. Thus, we have
\begin{align*}
\val(\Pi) \geq \val_{\Pi}(\psi) = \frac{1}{k \cdot M} \sum_{j \in [M]} \sum_{i \in [k]} \frac{|\psi_{y_j}^{-1}(i)|}{|\U|} &= \frac{1}{k \cdot M} \sum_{j \in [M]} \frac{|\U_j \cap (S_{\ell_1} \cup \cdots \cup S_{\ell_k})|}{|\U|} \\
&= \frac{1}{k \cdot M} \frac{|S_{\ell_1} \cup \cdots \cup S_{\ell_k}|}{|\U|} = \frac{\opt(\U, \CS)}{k \cdot M }.
\end{align*}
On the other hand, to show that $\val(\Pi) \leq \frac{\opt(\U, \CS)}{k \cdot M }$, let $\psi$ be any assignment of $\Pi$. We have
\begin{align*}
\val_{\Pi}(\psi) &\leq \frac{1}{k \cdot M} \sum_{j \in [M]} \sum_{i \in [k]} \frac{|\psi_{y_j}^{-1}(i) \cap S_{\psi_{x_i}}|}{|\U|} \\
&\leq \frac{1}{k \cdot M} \sum_{j \in [M]} \frac{|\U_j \cap (S_{\psi_{x_1}} \cup \cdots \cup S_{\psi_{x_k}})|}{|\U|} \\
&\leq \frac{1}{k \cdot M} \frac{|S_{\psi_{x_1}} \cup \cdots \cup S_{\psi_{x_k}}|}{|\U|} \leq \frac{\opt(\U, \CS)}{k \cdot M }. & \qedhere
\end{align*}
\end{proof}

\begin{remark}\label{rem:color}
    We can extend Lemma~\ref{lem:cov-to-vcsp}  to apply for the multicolored \mc problem as well in the following way. In particular, we can start from the  multicolored $\gapmaxcpv(\tau, (1 - \delta)\tau)$ problem and as described in Remark~\ref{rem:color0}, we can reduce it to the  multicolored $\gapmaxcpvbd(\tau', (1 - \eps)\tau')$ problem for $\eps = \delta^2/2$ and $\tau'=\delta (1-\delta)\cdot (1+\delta^2)$. Then, we note that we can mimic the proof of Lemma~\ref{lem:cov-to-vcsp} with one minor modification that for all $i\in[k]$, the alphabet set of variable $x_i$ are the indices of the $i^{\text{th}}$ collection of the input sets instead of the entire set $[n]$.
\end{remark}

\subsection{Step III: Valued \CSP $\Rightarrow$ 2-\CSP}

In this subsection, we prove Lemma~\ref{lem:vcsp-to-csp}. 

\begin{proof}[Proof of Lemma~\ref{lem:vcsp-to-csp}]
Let $\Pi = (V, E, (\Sigma_v)_{v \in V}, (f_e)_{e \in E})$ be a Valued 2-\CSP instance, and let $\ell = |E|$ and $\gamma = \ell \cdot s$.
We may assume that $\supp(f_e) \subseteq [0, \gamma)$ for all $e \in E$. Indeed, if any $\sigma_u \in \Sigma_u, \sigma_v \in \Sigma_v$ for $(u, v) \in E$ satisfies $f_{(u, v)}(\sigma_u, \sigma_v) \geq \ell \cdot s$, then assigning $\sigma_u$ to $u$ and $\sigma_v$ to $v$ alone already yields value at least $s$. We describe the reduction under this assumption. 

Let $B = \lceil 2 \ell / \eps \rceil$.
For each $\theta \in [B]^{|E|}$, check if $\frac{1}{|E|} \sum_{e \in E} \theta_e \geq B / \gamma \cdot \frac{s}{1-\eps}$. If not, then skip this $\theta$ and continue to the next one.
Otherwise, if this is satisfied, we create an instance $\Pi^\theta = (V, E, (\Sigma_v)_{v \in V}, (w^\theta_e)_{e \in E}, (C^\theta_e)_{e \in E})$ as follows:
\begin{itemize}
\item $V, E, (\Sigma_v)_{v \in V}$ remains the same as in $\Pi$.
\item For each $e = (u, v) \in E$, let $w^\theta_e = \theta_e$ and $C_e = \{(\sigma_u, \sigma_v) \mid f_e(\sigma_u, \sigma_v) \geq \gamma \cdot \theta_e / B\}$.
\end{itemize}

Note that the number of different $\theta$'s is $B^{|E|} = O(\ell / \eps)^{\ell} \leq 2^{O(\ell \log \ell)}$ and thus the above is an \FPT  reduction. We next prove the completeness and soundness of the reduction. 

\paragraph{Completeness.} Suppose that there is an assignment $\psi$ of $\Pi$ such that $\val_{\Pi}(\psi) \geq c$. Let $\theta^\psi$ be defined by $\theta^\psi_e := \lfloor B \cdot f_e(\psi_u, \psi_v) / \gamma \rfloor$ for all $e = (u, v) \in E$. Notice that
\begin{align*}
\frac{1}{|E|} \sum_{e \in E} \theta^\psi_e &= \frac{1}{|E|} \sum_{e = (u, v) \in E} \lfloor B \cdot f_e(\psi_u, \psi_v) / \gamma \rfloor \\
&\geq \frac{1}{|E|} \sum_{e = (u, v) \in E} \left(B \cdot f_e(\psi_u, \psi_v) / \gamma - 1\right) \\
&= B/\gamma \cdot \val_\Pi(\psi) - 1 \\
&\geq B/\gamma \cdot c - 1 \\
&\geq B / \gamma \cdot \frac{s}{1-\eps},
\end{align*}
where the last inequality follows from our choice of $B$ and $\eps$.

Thus, the instance $\Pi^{\psi}$ is considered in the construction. It is also obvious by the construction that $\Pi^{\psi}$ is indeed satisfiable.

\paragraph{Soundness.} Suppose (contrapositively) that for some $\theta$ with $\frac{1}{|E|} \sum_{e \in E} \theta_e \geq B / \gamma \cdot s$ such that $\Pi^\theta$ is not a NO instance of $\gapcsp(1, 1-\eps)$. That is, there exists an assignment $\psi$ such that $\val_{\Pi^\theta}(\psi) \geq 1-\eps$. From this, we have
\begin{align*}
\val_{\Pi}(\psi) &= \frac{1}{|E|} \sum_{e = (u, v) \in E} f_e(\psi_u, \psi_v) \\
&\geq \frac{1}{|E|} \sum_{e = (u, v) \in E} (\gamma \cdot \theta_e / B) \cdot \ind[(\psi_u, \psi_v) \in C_e] \\
&= \frac{1}{|E|} \gamma / B \cdot \val_{\Pi^\theta}(\psi) \cdot \left(\sum_{e \in E} \theta_e\right) \\
&\geq \frac{1}{|E|} \gamma / B \cdot (1 - \eps) \cdot \left(B / \gamma \cdot \frac{s}{1-\eps}\right) \\
&\geq s,
\end{align*}
where the first inequality is based on how $C^\theta_e$ is defined and the second inequality follows from $\val_{\Pi^\theta}(\psi) \geq 1-\eps$ and the assumption on $\theta$.

Thus, in this case, we have $\val(\Pi) \geq s$ as desired.
\end{proof}

\section{Reducing $k$-median to Multicolored $k$-MaxCoverage}

In this section, we prove the following formal version of Theorem~\ref{thm:kmedtomc}.

\begin{theorem}\label{thm:kmedtomcformal}\begin{sloppypar}
For every constant $\alpha,\delta>0$, there is an algorithm $\mathcal{A}$ which takes as input a $k$-median instance $((V,d),C,\mathcal{F},\tau)$  and outputs   $\Gamma(k)$ many multicolored \mc instances $\{(\mathcal{U}^i,\mathcal{S}^i:=\mathcal{S}_1^i\dot\cup \mathcal{S}_2^i\dot\cup \cdots \dot\cup \mathcal{S}_k^i)\}_{i\in [\Gamma(k)]}$, for some computable function $\Gamma:\N\to\N$, such that the following holds.\end{sloppypar}
\begin{description}
    \item[Running Time:] $\mathcal{A}$ runs in time $T(k)\cdot \poly(|V|)$, for some computable function $T:\N\to\N$. 
    \item[Size:] For every $i\in[\Gamma(k)]$, we have that $|\mathcal{U}^i|,|\mathcal{S}^i|=\poly(|V|\cdot \Delta)$, where $\Delta:=\frac{\underset{v,v'\in V}{\max}d(v,v')}{\underset{\tiny\substack{v,v'\in V\\ d(v,v')>0}}{\min}d(v,v')}$. 
    \item[Completeness:] Suppose that there exist $F\subseteq \mathcal{F}$ such that $|F|=k$ and $\cost(C,F)\le \tau$  then there exists $i\in [\Gamma(k)]$ and $(S_1,S_2,\ldots ,S_k)\in\mathcal{S}_1^i\times \mathcal{S}_2^i\times \cdots \times \mathcal{S}_k^i$ such that $S_1\cup S_2\cup \cdots \cup S_k=\mathcal{U}^i$. 
    \item[Soundness:]  Suppose that for every $F\subseteq \mathcal{F}$ such that $|F|=k$ we have $\cost(C,F)\ge (1+\alpha+\delta)\cdot \tau$  then for every $i\in[\Gamma(k)]$ and every $(S_1,S_2,\ldots ,S_k)\in\mathcal{S}_1^i\times \mathcal{S}_2^i\times \cdots \times \mathcal{S}_k^i$ we have  $|S_1\cup S_2\cup \cdots \cup S_k|\le (1-\frac{\alpha}{2})\cdot |\mathcal{U}^i|$. 
\end{description}

A similar reduction also holds from the $k$-means problem to the multicolored \mc problem.
\end{theorem}

Thus, from the above theorem we can show that a $(1 - 1/e - \eps)$-\FPT approximation for  the multicolored \mc problem implies a $(1 + 2/e + 3\eps)$-\FPT approximation for the $k$-median problem (by setting $\alpha=2\eps+(2/e)$ and $\delta=\eps$ in Theorem~\ref{thm:kmedtomcformal}). 
This reduction was almost established in Cohen-Addad et al.~\cite{cohen2019tight} who gave an $(1+2/e+\eps)$-approximation for $k$-median in time $(k \log k/\eps)^{O(k)} \poly(n)$, using a $(1-1/e)$-approximation algorithm for {\em Monotone Submodular Maximization with a (Partition) Matroid Constraint}: given a monotone submodular function $f : U \to \mathbb{R}_{\geq 0}$ and a partition matroid $(U, \calI)$, compute a set $S \in \calI$ that maximizes $f(S)$. 
Here we observe that the multicolored \mc problem can replace the general submodular maximization.

\paragraph{Algorithm of~\cite{cohen2019tight}.} First, we summarize the key steps from~\cite{cohen2019tight} without proofs. 
First, they compute a {\em coreset} of size $O(k \log n / \eps^2)$; it is a (weighted) subset of clients $\calC' \subseteq \calC$ such that for any solution $F \subseteq \calF$ with $|F| = k$, the \kmed costs for $\calC'$ and $\calC$ are within a $(1+\eps)$ factor of each other. Therefore, for the rest of the discussion, let $\calC$ be the coreset itself and assume $|\calC| = O(k \log n / \eps^2)$.

Let $F^* = \{ f^*_1, \ldots, f^*_k \}$ be the centers in the optimal solution, and $C^*_i$ be the set of clients served by $f^*_i$ in the optimal solution. For each $i \in [k]$, let $\ell_i \in C^*_i$ be the client closest to $f^*_i$ (ties broken arbitrarily) and call it the {\em leader} of $C^*_i$. 

By exhaustive enumerations, in $(k \log n / \eps)^{O(k)}$ time (which is upper bounded by $(k \log k/\eps)^{O(k)}\poly(n)$ time), one can guess $\ell_1, \ldots, \ell_k$ as well as $R_1, \ldots, R_k$ such that $d(\ell_i, f^*_i) \leq [R_i/(1+\eps), R_i]$. 

Let $F_i := \{ f \in \calF : d(f, \ell_i) \leq R_i \}$. (One can assume $F_i$'s are disjoint by duplicating facilities.) 

At this point, opening an arbitrary $f_i \in F_i$ for each $i \in [k]$ ensures a $(3+\eps)$-approximation, as each client $c \in C^*_i$ can be connected to $f_i$ where 
\[
d(c, f_i) \leq 
d(c, f^*_i) + d(f^*_i, \ell_i) + d(\ell_i, f_i)
\leq (3 + \eps)\cdot d(c, f^*_i), 
\]
since $d(f^*_i, \ell_i) \leq d(c, f^*_i)$ follows from the definition of the leader and 
$d(\ell_i, f_i) \leq (1+\eps)d(\ell_i, f^*_i)$ by the definition of $R_i$ and $F_i$. 

To further improve the approximation ratio to $(1 + 2/e + \eps)$,~\cite{cohen2019tight} defined the function $\improv : \mathcal{P}(\calF) \to \mathbb{R}_{\geq 0}$ as follows (here $\mathcal{P}$ denotes the power set). First, for each $i \in [k]$, add a {\em fictitious facility} $f'_i$, whose distance to $\ell_i$ is $R_i$ and the distances to the other points are determined by shortest paths through $\ell_i$; i.e., for any $x$, we define $d(f'_i, x) := R_i + d(\ell_i, x)$. 
Let $F' := \{ f'_1, \ldots, f'_k \}$. The above paragraph's reasoning again also shows that: 
\begin{align}
    \label{eq1}
\cost(C, F') \leq (3+\eps)\cdot \opt.
\end{align} 

For $S \subseteq \calF$, $\improv(S) := \cost(C, F') - \cost(C, S \cup F')$. 
Since any $f \in F_i$ is at a distance at most $R_i$ from $\ell_i$, as long as $|S \cap F_i| \geq 1$ for every $i \in [k]$, we have $\cost(C, S \cup F') = \cost(C, S)$. 

\cite{cohen2019tight} proved that $\improv(\cdot)$ is monotone and submodular, so one can use \cite{calinescu2011maximizing}'s algorithm which obtains an $(1-1/e)$-approximation algorithm for Monotone Subdmoular Maximization with a Matroid Constraint to find $S$ such that $|S \cap F_i| = 1$ for every $i \in [k]$ and $\improv(S) \geq (1-1/e)\cdot \improv(F^*)$, which implies that for an approximate solution $S^*$ we have:
\begin{align}
 \cost(C, S^*) &= 
\cost(C, F') - \improv(S^*)\nonumber\\
&\leq \cost(C, F') - (1-1/e)\cdot \improv(F^*) \nonumber\\
&\leq  \cost(C, F') - (1-1/e)\cdot (\cost(C, F') - \opt)\nonumber\\
&= (1-1/e)\cdot \opt + (1/e)\cdot \cost(C,F')\nonumber\\
&\leq (1 + 2/e + \eps) \cdot \opt.\label{eq}
\end{align}

\paragraph{$\improv(\cdot)$ as a Coverage Function.}
Since the matroid constraint exactly corresponds to the {\em multicolor} part of the multicolored \mc problem, it suffices to show that $\improv(\cdot)$ can be realized as a coverage function; it will imply that a $(1-1/e - \eps)$-approximation algorithm for multicolored \mc problem will imply $(1 + 2/e + O(\eps))$-approximation for $k$-median in \FPT time. 

Actually, the structure of $\improv(\cdot)$ as the cost difference between two $k$-median solutions makes it easy to do so. Let us do it  for a {\em weighted} coverage function where each element $e$ has weight $w(e)$ and the goal is to maximize the total weight of covered elements. (It can be unweighted via standard duplication tricks.) For each $c \in C$, let $d_c := d(c, F')$ and let $F_c := \{ f \in \calF : d(f, c) < d_c \}$. Let $F_c = \{ f_1, \ldots, f_t \}$ ordered in the decreasing order of $d(f_i, c)$. We create $t$ elements $E_c := \{ e_{c,1}, \ldots, e_{c, t} \}$ where $w(e_{c,1}) = d_c - d(f_1, c)$ and  
$w(e_{c,j}) = d(f_{j-1}, c) - d(f_j, c)$ for $j \in \{2, \ldots, t \}$. We will have a set $S_f$ for each $f \in \calF$, and for each $c \in C$, if $f = f_i$ in $F_c$'s ordering, $S_f \cap E_c = 
\{ e_{c,1}, \ldots, e_{c,i} \}$. (If $f \notin F_c$, $S_f \cap E_c \} = \emptyset$.)

Then for any $F \subseteq \calF$ and for any client $c \in C$, our construction ensures that $(\cup_{f \in F} S_f) \cap E_c$ is equal to $S_{f_c} \cap E_c$ where $f_c$ is the closest facility to $c$ in $F$, and the total weight of 
$(\cup_{f \in F} S_f) \cap E_c$ is exactly equal to $d(F', c) - d(f_c, c)$, which is exactly the improvement of the cost of $c$ in $F \cup F'$ compared to $F'$. 
Since $\improv(\cdot)$ is the sum of all clients, this function is a coverage function. 

\begin{proof}[Proof of Theorem~\ref{thm:kmedtomcformal}]
The algorithm $\mathcal{A}$ on input $((V,d),C,\mathcal{F},\tau)$ using the notion of {\em coresets} and exhaustive enumerations (and using randomness), in \FPT time constructs an instance for each guess of the values of $\ell_1, \ldots, \ell_k$ and $R_1, \ldots, R_k$.  The choice of $\eps$ in the use of coresets will be specified later. Thus, for a fixed instance (with $\ell_1, \ldots, \ell_k$ and $R_1, \ldots, R_k$ fixed), construct $F_c$ for all $c\in C$, and then the (weighted) set-system $\{S_f\}_{f\in\mathcal{F}}$ over the universe $\underset{c\in C}{\cup} E_C$. Then, following the exact same calculations as in \eqref{eq}, we have that for a solution $S^*\subseteq \mathcal{F}$ such that $\improv(S) \geq \left(1-\frac{\alpha}{2}\right)\cdot \improv(F^*)$, (for some $\alpha\ge 0$), we have:
\begin{align}
 \cost(C, S^*) &= 
\cost(C, F') - \improv(S^*)\nonumber\\
&\leq \cost(C, F') - \left(1-\frac{\alpha}{2}\right)\cdot \improv(F^*) \nonumber\\
&\leq  \cost(C, F') - \left(1-\frac{\alpha}{2}\right)\cdot (\cost(C, F') - \opt)\nonumber\\
&= \left(1-\frac{\alpha}{2}\right)\cdot \opt + \frac{\alpha}{2}\cdot \cost(C,F')\nonumber\\
&\le  \left(1-\frac{\alpha}{2}\right)\cdot \opt + \frac{\alpha}{2}\cdot (3+\eps)\cdot \opt= \left(1+\alpha+\frac{\alpha\eps}{2}\right)\cdot \opt,\nonumber 
\end{align}
where the last inequality follows from \eqref{eq1}. 
The theorem statement completeness and soundness claims then follows by choosing $\eps=\frac{2\delta}{\alpha}$. Moreover, the reduction is clearly in \FPT time, and the weights of the elements we constructed are bounded by $\Delta$, which is the blowup that happens in the size of the set system to reduce to the unweighted multicolored \mc problem.

\begin{remark}\label{rem:kmean}
The theorem statement also holds for the $k$-means objective as \eqref{eq1} is revised to $\cost_2(C, F') \leq (9+\eps)\cdot \opt$ and we can thus conclude that for a solution $S^*\subseteq \mathcal{F}$ such that $\improv(S) \geq \left(1-\frac{\alpha}{8}\right)\cdot \improv(F^*)$,  we have $\cost_2(C, S^*)\le \left(1-\frac{\alpha}{8}\right)\cdot \opt + \frac{\alpha}{8}\cdot (9+\eps)\cdot \opt= \left(1+\alpha+\frac{\alpha\eps}{8}\right)\cdot \opt$. 
\end{remark}

\end{proof}

\begin{remark}\label{rem:mc}
Starting from the multicolored \mc problem, we can apply Theorem~\ref{thm:mctocspformal} to obtain a gap preserving reduction to 2-\CSP, and then simply apply Lemma~19 in \cite{CGLL19} to obtain a gap preserving reduction to (uncolored) \mc problem. However, at the moment we do not know how to directly reduce the multicolored \mc problem to the uncolored version (while retaining the gap), but this convoluted procedure suggests that a direct reduction might be plausible. 
\end{remark}

At this point one may wonder if it is possible to provide a gap preserving \FPT reduction from the multicolored \mc problem to the (unmulticolored) \mc problem. Such a reduction would help us avoid Remark~\ref{rem:color} and directly compose the results of Theorems~\ref{thm:kmedtomcformal}~and~\ref{thm:mctocspformal} in a blackbox manner. While reductions from the multicolored variant of a problem to its unmulticolored counterpart can be quite straightforward, such as for the $k$-clique and $k$-set cover problem, it can also be quite notorious such as for the $k$-set intersection problem \cite{karthik2020closest,BukhSN21}. Such reductions are aptly dubbed as ``reversing color coding''. To the best of our knowledge, reversing color coding for the \mc problem can be quite hard (if we want to preserve some positive constant gap). That said,  Remark~\ref{rem:mc} does provide a convoluted argument as to while a direct reversing color coding is currently out of reach, it can still be plausibly achieved. We note here that if we are promised that the multicolored \mc problem instance has all the input sets of roughly the same size then there is a simple reduction to the (uncolored) \mc problem by simply introducing $k$ new subuniverses, one for each color class.

\section{Inapproximability of $k$-MaxCoverage}
In this section, we prove (the \W$[1]$-hardness part of) Theorem~\ref{thm:mccomplete}.
In \cite{KLM19}, the authors had \emph{implicitly} proved the following: for some computable function $F$, it is \W$[1]$-hard given $(\U,\CS)$ to find $k$ sets in $\CS$  that cover $1-\nicefrac{1}{F(k)}$ fraction of $\U$ whenever there exists  $k$ sets in $\CS$ that cover $\U$.  We show below how that can be extended to obtain the more generalized result given in Theorem~\ref{thm:mccomplete}.

Our proof builds on the following \W$[1]$-hardness of gap $k$-\maxcover proved in \cite{KLM19}.

\begin{theorem}[\cite{KLM19}]\label{thm:ETHmax}
There exists a computable  function $A\colon \mathbb{N}\to\mathbb{N}$ such that it is \W$[1]$-hard to decide an instance $\Gamma=\left(G=(V\dot\cup W,E),1,\frac{1}{2}\right)$ of   $k$-\maxcover  even in the following setting:
\begin{itemize}
    \item $V:=V_1\dot\cup\cdots\dot\cup V_k$, where $\forall j\in[k]$, $|V_j|=n$.
    \item $W:=W_1\dot\cup\cdots\dot\cup W_\ell$, where $\ell=(\log n)^{O(1)}$ and $\forall i\in[k]$, $|W_i|=A(k)$.
\end{itemize}
\end{theorem}
\begin{proof}[Proof Sketch]
All the references here are using the labels in~\cite{KLM19}. First we apply Proposition 5.1 to Theorem 7.1 with $z=\frac{1}{\log_2 \left(\frac{1}{1-\delta}\right)}$ to obtain a $(0,O(\log_2 m),O(t),1/2)$-efficient protocol for $k$-player $\mathsf{MultEq}_{m,k,t}$ in the SMP model. 
The proof of the theorem then follows by plugging in the parameters of the protocol  to Corollary 5.5. 
\end{proof}

\begin{sloppypar}
Starting from the above theorem, we   mimics ideas from 
Feige's proof of the \NP-hardness of approximating the {Max Coverage} 
problem \cite{feige1998threshold}. Let $T:=T(k)$ be an integer such that $\rho(Tk)\ge 2\cdot k^{A(k)}$ (such an integer $T$ exists because $\rho$ is unbounded).
Given an instance $\Gamma(G=(V\dot\cup W,E),1,1/2)$ of   $k$-\maxcover, we construct the universe $\U:=\{(t,i,f):t\in[T],i\in [\ell],\ f:W_i\to [k]\}$, and a collection of sets $\CS:=\{S_{(t,j,v)}\}_{t\in[T], j\in [k], v\in V_j}$, where for all $t\in [T]$ we have $(t,i,f)\in S_{(t,j,v)}\iff \exists w\in W_i$ such that $f(w)=j$ and $(v,w)\in E$. Note that $|\U|=T\cdot (\log n)^{O(1)} \cdot k^{A(k)}$ and that $|\CS|=T\cdot k\cdot n$. 
\end{sloppypar}

Suppose there exists  $v_1\in V_1,\ldots ,v_k\in V_k$ such that for all $i\in[\ell]$ we have that $W_i$ is covered by $v_1,\ldots ,v_k$. Let $w_i\in W_i$ be a common neighbor of  $v_1,\ldots ,v_k$. Then, we claim that the collection $\{S_{(t,1,v_1)},S_{(t,2,v_2)},\ldots ,S_{(t,k,v_k)}\}_{t\in [T]}$ covers $\U$. 
This is because for every $t\in[T]$ and every $(t,i,f)\in \U$  we have that $S_{(t,f(w_i),v_{f(w_i)})}$ covers it.

On the other hand, suppose that for every  $v_1\in V_1,\ldots ,v_k\in V_k$  we have that only $1/2$ fraction of the $W_i$s are covered by $v_1,\ldots ,v_k$. Fix some $Tk$ sets $\tilde{\CS}$ in $\CS$. 
For every $t\in [T]$, let $\U_t:=\{(t,i,f):i\in [\ell], f:W_i\to [k]\}$ and $\tilde{\CS}_t:=\tilde{\CS}\cap \{S_{(t,j,v)}:j\in [k], v\in V_j\}$. We can partition $\tilde{\CS}$ to $\tilde{\CS}^-,\tilde{\CS}^=,$ and $\tilde{\CS}^+$, where for every $t\in [T]$ we include all the sets in $\tilde{\CS}_t$ to $\tilde{\CS}^-$ if $|\tilde{\CS}_t|<k$, to $\tilde{\CS}^+$ if $|\tilde{\CS}_t|>k$, and to $\tilde{\CS}^=$ if $|\tilde{\CS}_t|=k$.
Let $R^+\subseteq [T]$ (resp.\ $R^-\subseteq [T]$) be defined as follows: $t\in R^+\iff |\tilde{\CS}_t|>k$ (resp.\ $t\in R^-\iff |\tilde{\CS}_t|<k$). Since $|\tilde{\CS}|=Tk$, we  have that there exists $Q\subseteq [T]\times [k]$ such that $|Q|= |\tilde{\CS}^+|-(|R^+|\cdot k)$ and $(t,j)\in Q\iff \tilde{\CS}_t\cap \{S_{(t,j,v)}: v\in V_j\}=\emptyset$ and $|\tilde{\CS}_t|<k$.

Now, we observe that once we fix $t\in [T]$ and $j\in [k]$ and  if $\tilde{\CS}_t\cap \{S_{(t,j,v)}: v\in V_j\}=\emptyset$ then $\tilde{\CS}$ does not cover any element in the subuniverse $\{(t,i,f_j):i\in [\ell]\}$, where $f_j$ is the constant function which maps everything to $j$. Therefore, we can conclude that there are $|Q|\cdot \ell$ many elements in $\underset{t\in R^-}{\cup}\U_t$ such that are not covered by $\tilde{\CS}$. 

Also, suppose we picked $S_{(t,1,v_1)},S_{(t,2,v_2)},\ldots ,S_{(t,k,v_k)}$ for some $t\in[T]$, $v_1\in V_1,\ldots ,v_k\in V_k$ then for every $i\in [\ell]$ such that $v_1,\ldots ,v_k$ does not cover $W_i$, we have that for every $w\in W_i$ there is some $j\in[k]$ such that $(v_j,w)\notin E$. Therefore, there is some $f:W_i\to[k]$ such that $S_{(t,1,v_1)},S_{(t,2,v_2)},\ldots ,S_{(t,k,v_k)}$ does not cover $(t,i,f)$. Thus, these $k$ sets do not cover at least $\ell/2$ universe elements.

We can now put everything together to complete the soundness analysis. First note that for every $t\in [T]$, no element of the subuniverse $\U_t$ can be covered by a set in $\tilde{\CS}\setminus \tilde{\CS}_t$. For every $t\in [T]$, if $|\tilde{\CS}_t|=k$ then from the above analysis we have that at least $\ell/2$ elements of $\U_t$ are not covered by $\tilde{\CS}$. 

Therefore, the total number of universe elements not covered by $\tilde{\CS}$ is at least $\left(|\tilde{\CS}^=|\cdot \ell/2k\right) + |Q|\cdot \ell=\ell\cdot \left(|\tilde{\CS}^=|/2k+|\tilde{\CS}^+|-(|R^+|\cdot k)\right)$. Since $|\tilde{\CS}^=|=(T-R^+-R^-)\cdot k$, this implies that $\tilde{\CS}$ does not cover at least $\ell\cdot \left(T/2+|\tilde{\CS}^+|-(|R^+|\cdot k)-\frac{R^++R^-}{2}\right)$ elements in $\U$. Next, we note that by  the way we  partitioned  $\tilde{\CS}$, we have $|\tilde{\CS}^+|-(|R^+|\cdot k)=(|R^-|\cdot k)-|\tilde{\CS}^-|$, and we also have $|\tilde{\CS}^+|\ge |R^+|\cdot (k+1)$ and $|\tilde{\CS}^-|\le |R^-|\cdot (k-1)$. From this we can surmise that  $|\tilde{\CS}^+|-(|R^+|\cdot k)\ge \max(|R^+|,|R^-|)\ge \frac{R^++R^-}{2}$. Thus, we have that  $\tilde{\CS}$ does not cover at least $\ell T/2$ universe elements.

Thus, we can conclude that $\tilde{\CS}$ can not cover at least $T\ell/2$ universe elements. This is $k^{-A(k)}/2$ fraction of $\U$ that is not covered by $\tilde{\CS}$. The proof follows by noting that $k^{-A(k)}/2\ge 1/\rho(Tk)=1/\rho(|\tilde{\CS}|)$.

\subsection*{Acknowledgement}
 We are grateful to the Dagstuhl Seminar 23291 for a special collaboration opportunity.
We thank Vincent Cohen-Addad, Venkatesan Guruswami, Jason Li, Bingkai Lin, and Xuandi Ren for discussions.

\bibliographystyle{alpha}
\bibliography{ref}

\end{document}